\documentclass[11pt]{article}
\usepackage{style}

\title{New lower bounds for CDS and $f$-routing}

\author[1]{Atsuya Hasegawa\thanks{atsuya.hasegawa@math.nagoya-u.ac.jp}}
\author[1]{Ranitha Mataraarachchi\thanks{ranitha@nagoya-u.jp}}
\affil[1]{\textit{Graduate School of Mathematics, Nagoya University, Japan}}

\date{}

\begin{document}

\maketitle

\begin{abstract}
Understanding the entanglement cost of non-local quantum computation (NLQC) is relevant to complexity theory, cryptography, quantum gravity, and other areas.
A central special case is $f$-routing, motivated in part by quantum position verification. Proving lower bounds on its entanglement cost in the fully robust setting has been a major open problem in NLQC.

Motivated by this problem, we establish two related lower bounds. First, we study the shared-randomness cost of robust conditional disclosure of secrets (CDS). The connection between CDS and $f$-routing established by Allerstorfer et al. (Quantum 2024) makes understanding the randomness complexity of robust CDS a natural step toward lower bounds for the fully robust routing problem. We show that the shared-randomness cost of robust CDS is lower bounded by the logarithm of deterministic SMP communication complexity, even when communication and private randomness are unrestricted. Our lower bound is tight for the equality function.

Second, we consider one-sided-perfect $f$-routing, in which the protocol is exact on one input class and has constant error on the other. By exploiting the positivity of the low-rank matrix in the method of Asadi, Culf, and May (ITCS 2025), we derive a lower bound on the entanglement cost in terms of sign rank. In particular, this yields a linear lower bound on the entanglement cost of routing for the inner-product function in both one-sided-perfect settings, matching the known upper bound.
 
\end{abstract}

\clearpage

\tableofcontents

\clearpage

\section{Introduction}

\subsection{Background}
\paragraph*{NLQC and $f$-routing.}
A bipartite quantum operation can be implemented directly by bringing the two input systems together and allowing them to interact. In \emph{non-local quantum computation} (NLQC), however, the systems remain spatially separated and cannot interact directly. Instead, the parties implement the desired operation using local operations, a pre-shared entangled state, and a single simultaneous round of quantum communication. See \cref{fig:frouting-task} for an illustration. A central question in NLQC is how much entanglement is required to implement a given operation in this form. NLQC is related to a wide range of fields, including quantum position verification (QPV)~\cite{kent2011quantum,buhrman2014position,beigi2011simplified}, the AdS/CFT correspondence~\cite{may2019quantum,may2020holographic,may2022complexity}, Hamiltonian simulation~\cite{apel2024security}, and information-theoretic cryptography~\cite{allerstorfer2024relating,
asadi2025conditional}, computational and communication complexity \cite{buhrman2013garden,speelman2016t-depth,girish2026magic}. See also \cite{may2026entanglement} for a recent survey on NLQC.

A particularly interesting family of NLQC tasks is \emph{$f$-routing}. Alice receives a classical input $x$ and an unknown quantum system $Q$, while Bob receives a classical input $y$. For a fixed Boolean function $f:X\times Y\to\{0,1\}$, they must ensure that Alice can recover $Q$ when $f(x,y)=0$, while Bob can recover $Q$ when $f(x,y)=1$. The motivation for this task is rooted in quantum position verification (QPV) \cite{kent2011quantum,buhrman2014position,beigi2011simplified}. An honest prover (who acts locally) can evaluate $f(x,y)$ classically and then redirect the quantum system $Q$ according to the value of $f$. Thus, when $Q$ has fixed dimension, the honest prover requires only $O(1)$ quantum operations. By contrast, spatially separated cheating agents must implement an $f$-routing protocol. One main goal in $f$-routing is to prove that the entanglement required by such agents grows with the input size. This is precisely the type of separation desired in QPV: the honest prover performs an almost entirely classical computation, whereas successful cheating agents could require a large shared quantum resource.

\begin{figure}
    \centering
    \begin{subfigure}{0.45\textwidth}
    \centering
    \begin{tikzpicture}[scale=0.6]
    
    \draw[thick] (-1,-1) -- (-1,1) -- (1,1) -- (1,-1) -- (-1,-1);
    
    \draw[thick] (-3.5,-3) to [out=90,in=-90] (-0.5,-1);
    \node[below] at (-3.5,-3) {$A$};
    \draw[thick] (3.5,-3) to [out=90,in=-90] (0.5,-1);
    \node[below] at (3.5,-3) {$B$};
    
    \draw[thick] (0.5,1) to [out=90,in=-90] (3.5,3);
    \node[above] at (3.5,3) {$B$};
    \draw[thick] (-0.5,1) to [out=90,in=-90] (-3.5,3);
    \node[above] at (-3.5,3) {$A$};
    
    \node at (0,0) {$\mathcal{T}$};

    \node at (0,-5) {$ $};
    
    \end{tikzpicture}
    \caption{}
    \label{fig:local}
    \end{subfigure}
    \hfill
    \begin{subfigure}{0.45\textwidth}
    \centering
    \begin{tikzpicture}[scale=0.4]
    
    \draw[thick] (-5,-5) -- (-5,-3) -- (-3,-3) -- (-3,-5) -- (-5,-5);
    \node at (-4,-4) {$\mathcal{V}^L$};
    
    \draw[thick] (5,-5) -- (5,-3) -- (3,-3) -- (3,-5) -- (5,-5);
    \node at (4,-4) {$\mathcal{V}^R$};
    
    \draw[thick] (5,5) -- (5,3) -- (3,3) -- (3,5) -- (5,5);
    \node at (4,4) {$\mathcal{W}^R$};
    
    \draw[thick] (-5,5) -- (-5,3) -- (-3,3) -- (-3,5) -- (-5,5);
    \node at (-4,4) {$\mathcal{W}^L$};
    
    \draw[thick] (-4.5,-3) -- (-4.5,3);
    
    \draw[thick] (4.5,-3) -- (4.5,3);
    
    \draw[thick] (-3.5,-3) to [out=90,in=-90] (3.5,3);
    
    \draw[thick] (3.5,-3) to [out=90,in=-90] (-3.5,3);
    
    \draw[thick] (-3.5,-5) to [out=-90,in=-90] (3.5,-5);
    \draw[black] plot [mark=*, mark size=3] coordinates{(0,-7.05)};
    
    \draw[thick] (-4.5,-6) -- (-4.5,-5);
    \node[below] at (-4.5,-6) {$A$};
    \draw[thick] (4.5,-6) -- (4.5,-5);
    \node[below] at (4.5,-6) {$B$};
    
    \draw[thick] (4.5,5) -- (4.5,6);
    \node[above] at (-4.5,6) {$A$};
    \draw[thick] (-4.5,5) -- (-4.5,6);
    \node[above] at (4.5,6) {$B$};
    
    \end{tikzpicture}
    \caption{}
    \label{fig:non-localcomputation}
    \end{subfigure}
    \caption{Local and non-local computations. a) A channel $\mathcal{T}_{AB\rightarrow AB}$ is implemented by directly interacting the input systems. b) A non-local quantum computation. The goal is for this circuit's action on the $AB$ systems to approximate the channel $\mathcal{T}_{AB\rightarrow AB}$. Figure reproduced from~\cite{asadi2025rank}.}
    \label{fig:frouting-task}
\end{figure}
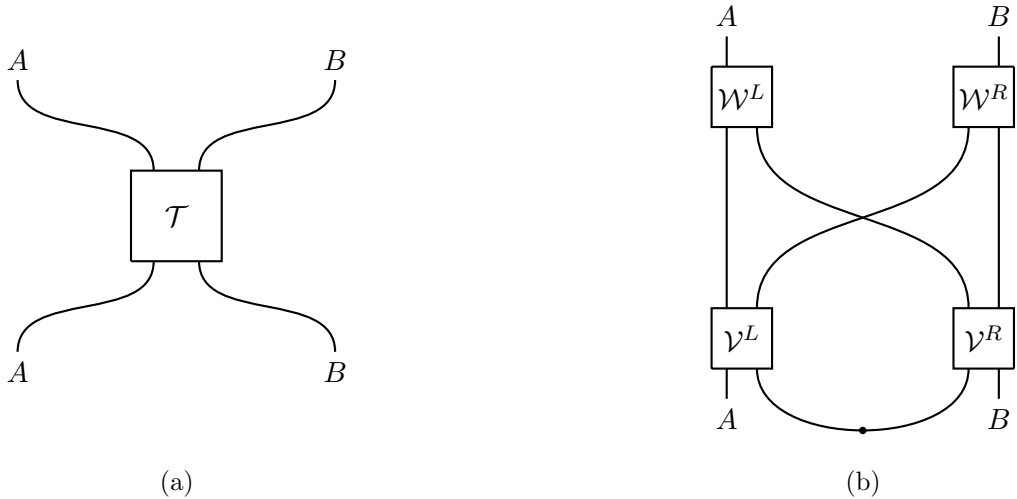

Every \(f\)-routing task on two \(n\)-bit inputs can be completed with entanglement cost \(2^{O(\sqrt{n\log n})}\) \cite{allerstorfer2024relating}.\footnote{A simpler argument using the garden-hose construction yields a $2^{O(n)}$ upper bound~\cite{buhrman2013garden}.} In contrast, entanglement lower bounds for \(f\)-routing remain poorly understood. Asadi, Culf, and May~\cite{asadi2025rank} obtained lower bounds on the amount of the pre-shared entangled state. However, their bounds apply only to one-sided perfect protocols, where recovery must be perfect either on every \(0\)-input or on every \(1\)-input. Physical implementations of $f$-routing protocols inevitably involve noise and imperfect operations, and thus one generally expects a small probability of error on both input classes. A lower bound that requires perfect recovery on one class therefore does not directly apply to such implementations. Obtaining entanglement lower bounds for robust $f$-routing, where constant error is allowed on both input classes, has been a major open problem.\footnote{In independent work, Bogner~\cite{bogner2026robust} recently showed an entanglement lower bound for $f$-routing with the inner product function on the two-sided error setting. See also \cref{subsec:concurrent}.} See also Chapter 14 of \cite{may2026entanglement}.

\paragraph*{Conditional disclosure of secrets (CDS).}
In this paper, we also consider a classical information theoretic primitive known as \emph{conditional disclosure of secrets} (CDS)~\cite{gertner2000cds}. In a CDS protocol, Alice receives an input \(x\) and a secret \(s\), Bob receives an input \(y\), and they share a random variable \(R\). Without communicating with one another, they send messages \(A\) and \(B\) to a referee who knows \((x,y)\). If \(f(x,y)=1\), the referee should recover \(s\); if \(f(x,y)=0\), the pair of messages should reveal essentially no information about \(s\). See
\cref{fig:cds-task} for an illustration. The shared randomness allows Alice and Bob to coordinate their messages so that the secret is recoverable on the \(1\)-inputs while remaining hidden on the \(0\)-inputs. A question is how much shared randomness is required to satisfy both conditions.

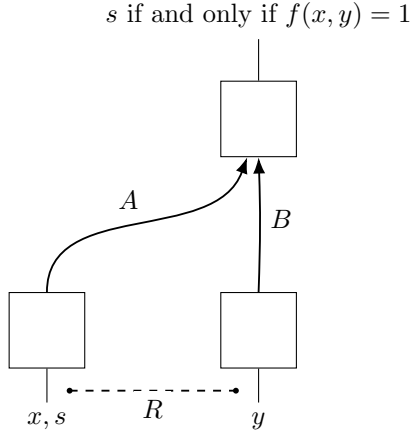
\begin{figure}[t]
    \centering
    \begin{tikzpicture}[
        >=Latex,
        box/.style={draw, minimum size=10mm, inner sep=0pt},
        lab/.style={font=\small, align=center},
        msg/.style={->, line width=.7pt}
        ]
        \node[box] (alice) at (-1.6,0) {};
        \node[box] (bob) at (1.2,0) {};
        \node[box] (referee) at (1.2,2.8) {};

        \draw (alice.south) -- +(0,-0.45)
            node[lab, below] {$x,s$};
        \draw (bob.south) -- +(0,-0.45)
            node[lab, below] {$y$};

        \draw[msg] (alice.north) to[out=90,in=250]
            node[lab, above left, pos=.55] {$A$} ([xshift=-0.15cm]referee.south);
        \draw[msg] (bob.north) to[out=88,in=272]
            node[lab, right, pos=.55] {$B$} (referee.south);

        \coordinate (leftR) at (-1.30,-0.80);
        \coordinate (rightR) at (0.90,-0.80);
        \fill (leftR) circle (1.1pt);
        \fill (rightR) circle (1.1pt);
        \draw[dashed, line width=.7pt] (leftR) -- (rightR)
            node[lab, fill=white, inner sep=3pt, midway, below] {$R$};

        \draw (referee.north) -- +(0,0.55)
            node[lab, above] {$s$ if and only if $f(x,y)=1$};
    \end{tikzpicture}
    \caption{The CDS task. Alice (left) receives $(x,s)$ and Bob (right)
    receives $y$. Using a shared random variable $R$, they send messages $a$ and $b$ to the referee (top), who recovers $s$ if and only if $f(x,y)=1$.}
    \label{fig:cds-task}
\end{figure}

CDS is connected to NLQC through conditional disclosure of quantum secrets (CDQS), a quantum analogue of CDS in which the secret and messages may be quantum and shared randomness is replaced by entanglement~\cite{allerstorfer2024relating,asadi2025conditional}. A classical CDS protocol can be converted into a CDQS protocol, which can in turn be converted into an \(f\)-routing protocol with controlled overhead in the resource cost and error parameters~\cite{allerstorfer2024relating}, and the entanglement cost upper bound \(2^{O(\sqrt{n\log n})}\) \cite{allerstorfer2024relating} was shown from the randomness upper bound of CDS for any function $f$ in \cite{liu2017conditional}. This conversion direction allows entanglement lower bounds for \(f\)-routing to yield randomness lower bounds for CDS as shown in \cite{asadi2025rank}. The converse does not follow from these reductions; therefore, a randomness lower bound for CDS does not by itself imply an entanglement lower bound for \(f\)-routing. Nevertheless, robust CDS retains the same function-dependent resource question in a simpler classical setting, making it a natural place to develop ideas toward robust \(f\)-routing lower bounds.

Applebaum and Vasudevan~\cite{applebaum2021placing} showed a lower bound on shared randomness for perfectly correct CDS protocols, and Asadi, Culf and May~\cite{asadi2025rank} showed one for perfectly correct or secure CDS protocols. However, obtaining shared-randomness lower bounds that tolerate constant error in both correctness and privacy remained open in the model with unrestricted communication and unlimited free private randomness.

\subsection{Our results}

We prove two related lower bounds for the goal obtaining the error-robust entanglement lower bound for $f$-routing. 

\paragraph*{Robust lower bounds for CDS.}
Consider an \(\epsilon\)-correct and \(\delta\)-private CDS protocol for \(f\) whose shared random variable \(R\) has support size \(q\). We place no restrictions on the parties' private randomness or on the lengths of the messages sent to the referee. Our resource measure is the shared-randomness cost \(\log_2 q\).\footnote{\(\log_2 q\) corresponds to the number of shared random bits.} See \cref{def:robust-cds} for a formal definition for CDS protocols.

Let \(\D^{\|}(f)\) denote the deterministic communication complexity of \(f\) in the simultaneous message passing (SMP) model. Our result for robust CDS is as follows.

\begin{theorem}
[Informal version of \cref{thm:cds-randomness-lower-bound}]
\label{thm:informal1}
Let \(f:X\times Y\to\{0,1\}\). Suppose that \(f\) admits an \(\epsilon\)-correct and \(\delta\)-private CDS protocol where \(\epsilon,\delta>0\) are fixed constants satisfying $2\epsilon+\delta<1$.
Then,
\[
    \log_2 q
    =
    \Omega\left(\log_2\D^{\|}(f)\right).
\]
\end{theorem}

In particular, the deterministic SMP communication complexity of \(n\)-bit equality is \(\Theta(n)\). Hence~\cref{thm:informal1} gives an \(\Omega(\log n)\) shared-randomness lower bound for the equality function in the robust setting, which is tight up to constant factors. In Appendix~\ref{app:EQ_upperbound}, we give a perfectly correct $\delta$-private protocol for equality using $O\!\left(\log n\right)$ shared random bits for a constant $\delta>0$. 

\paragraph*{Sign-rank lower bounds for one-sided-perfect
\(f\)-routing.}
Our second result revisits the rank lower-bound method for one-sided \(f\)-routing introduced by Asadi, Culf, and May~\cite{asadi2025rank}. A one-sided perfect $f$-routing protocol uses a shared pure state of Schmidt rank \(d_E\), is exact on the \(0\)-inputs, and has error at most \(0.05\) on the \(1\)-inputs. The symmetric setting, with the roles of the two input classes reversed, is defined analogously. We write \(\FRzero(f)\), \(\FRone(f)\), and \(\pFR(f)\) for the minimum value of \(\log_2 d_E\) over one-sided and two-sided perfect protocols.

Let \(S_f\) be the sign matrix defined by
\begin{equation}\label{eq:signmatrix}
     S_f(x,y):=2f(x,y)-1.
\end{equation}
We obtain lower bounds for one-sided $f$-routing in terms of the sign rank of \(S_f\).

\begin{theorem}[Informal version of \cref{thm:main}]\label{thm:informal2}
\[
    \FRzero(f),\ \FRone(f),\ \pFR(f)
    \ge
    \frac14
    \log_2\!\left(
        \max\{1,\signrank(S_f)-1\}
    \right).
\]
\end{theorem}

As a consequence of the relationship between the sign rank and the unbounded-error communication complexity~\cite{paturi1986probabilistic}, we also obtain
\[
    \FRzero(f),\ \FRone(f),\ \pFR(f)
    \ge
    \frac14\UPPcc(f)-O(1),
\]
where $\UPPcc(f)$ is the unbounded-error communication complexity of $f$.

For the modulo-\(2\) inner-product function $\IP_n$, the sign-rank lower bound in~\cite{forster2002linear} yields
\[
    \FRzero(\IP_n),\ \FRone(\IP_n),\ \pFR(\IP_n)
    \ge
    \frac{n}{8}-\frac{1}{4}.
\]
The lower bound in Asadi, Culf and May is characterized by the nondeterministic rank of $f$, and their analysis  yielded only a constant lower bound for this function~\cite[Figure~3]{asadi2025rank}.
Our lower bounds match the $O(n)$ upper bound obtained from the garden-hose construction, which gives a perfectly correct routing protocol for $\IP_n$~\cite{buhrman2013garden}. The two lower bounds behave differently on concrete functions. Their argument, for example, gives a linear lower bound for equality, whereas equality has constant sign rank (e.g., Remark 1.12 in \cite{hatami2023borsuk}). Thus, our sign-rank formulation does not replace the original argument in~\cite{asadi2025rank}, but provides a different way to derive lower bounds.

Finally, combining~\cref{thm:informal2} with known reductions~\cite{allerstorfer2024relating,bluhm2026complexity} yields corresponding lower bounds in terms of sign rank and unbounded-error communication complexity for one-sided-perfect \(f\)-BB84 and for CDS with perfect privacy or perfect correctness (\cref{cor:fbb84,cor:cds}).

\subsection{Proof overview}

\paragraph*{Robust CDS lower bound.}
The proof relates the number of distinct rows of \(f\) to the support size of the shared random variable $R$. For each input pair \((x,y)\), let \(T_s^{x,y}\) be the transcript distribution when the secret is \(s\), and define
\[
    D(x,y)
    :=
    \operatorname{TV}\bigl(T_0^{x,y},T_1^{x,y}\bigr).
\]
Privacy and correctness give
\[
    f(x,y)=0
    \ \Longrightarrow
    D(x,y)\le\delta,
    \qquad
    f(x,y)=1
    \ \Longrightarrow
    D(x,y)\ge 1-2\varepsilon.
\]
Thus, if \(x\) and \(x'\) define distinct rows of \(f\), then for some
\(y\),
\[
    |D(x,y)-D(x',y)|
    \ge
    \gamma,
    \qquad
    \gamma:=1-2\varepsilon-\delta.
\]

We next use the assumption that the shared random variable $R$ has support size $q$. Fix a CDS protocol satisfying the stated correctness and privacy conditions. Conditioning on Bob’s message produces a posterior distribution $p\in\Delta_q$ on $R$. 
For each Alice input \(x\), we define a function \(\phi_x(p)\) that measures how distinguishable
Alice's messages are for the two possible secrets when the posterior on \(R\) is \(p\). This gives the representation
\[
    D(x,y)
    =
    \mathbb E_{p\sim\mu_y}\bigl[\phi_x(p)\bigr],
\]
where $\mu_y$ is the distribution of the posterior on $R$ induced by Bob's encoder on input $y$.

To approximate these functions uniformly, we use the fact that each $\phi_x$ is a sum of absolute values of linear forms. Although the number of terms can depend on the message alphabet, their total coefficient mass is bounded solely in terms of $q$. Sampling and rounding these terms therefore yield a finite family of uniform approximations whose size depends only on $q$ and the approximation error. Our main technical lemma constructs a finite family \(\mathcal H\) such that every \(\phi_x\) is \(\eta\)-close to some \(h_x\in\mathcal H\) in the $\infty$-norm, with
\[
    \log_2|\mathcal H|
    \le
    c\frac{q^4}{\eta^2}
    \log_2^2\left(\frac{cq}{\eta}\right),
\]
for some universal constant $c$. Importantly, this bound is independent of the parties' private randomness and message lengths.

If two inputs \(x\) and \(x'\) have the same approximation
\(h_x=h_{x'}\), then
\[
    |D(x,y)-D(x',y)|
    \le 2\eta
    \qquad
    \text{for every }y.
\]
Taking \(\eta=\gamma/4\), this is impossible when \(x\) and \(x'\) define distinct rows of \(f\). Hence the number of distinct rows of \(f\) is at most \(|\mathcal H|\). Applying the analogous argument to the columns gives the same bound on the number of distinct columns. The numbers of different columns and rows are related to the deterministic SMP communication complexity, and thus we obtain the result.

\paragraph*{Sign-rank lower bound for \(f\)-routing.}
Consider first a protocol that is exact on the \(0\)-inputs. The rank method of Asadi, Culf, and May~\cite{asadi2025rank} associates with the protocol a real matrix \(G\) satisfying
\[
    G_{x,y}=0 \quad\text{when }f(x,y)=0,
    \qquad
    G_{x,y}>0 \quad\text{when }f(x,y)=1,
\]
and \(\rank(G)\le d_E^4\). For nonconstant \(f\), choose \(\tau>0\) smaller than every positive entry of \(G\), let \(J\) be the all-ones matrix, and set
\[
    A:=G-\tau J.
\]
Then \(A\) is negative on the \(0\)-inputs and positive on the \(1\)-inputs. Thus \(A\) has sign pattern \(S_f\), where $S_f$ was defined as in \cref{eq:signmatrix}. Consequently,
\[
    \signrank(S_f)
    \le \rank(A)
    \le \rank(G)+1
    \le d_E^4+1.
\]
The case in which the protocol is exact on the \(1\)-inputs follows by applying the same argument to \(1-f\). Taking logarithms gives the claimed routing-cost bound.

In Appendix \ref{appendix_nrank-signrank}, we show the same asymptotic lower bounds by combining the nondeterministic-rank lower bound of Asadi, Culf, and May~\cite{asadi2025rank} with a relation between nondeterministic rank and sign rank.
This approach yields a coefficient of $1/8$ in the sign-rank bound.
By directly exploiting the nonnegativity of $G$, our argument improves this coefficient to $1/4$, strengthening the lower bound for inner product from $n/16-1/8$ to $n/8-1/4$.
We believe that such constant-factor improvements are important in QPV, where the concrete number of shared entangled qubits required by cheating strategies matters.

\subsection{Related work}

\paragraph*{Conditional disclosure of secrets.}
Previous work~\cite{gay2015communication,applebaum2021conditional, applebaum2021placing} showed lower bounds on the communication cost of robust CDS. Lower bounds on shared randomness were also known under the one-sided perfect condition~\cite{applebaum2021placing,kawachi2021randomness,asadi2025rank}. More recently, Girish, May, Orshansky, and Waddell~\cite{girish2026comparing} proved robust lower bounds for quantum CDS, and the complexity measure of their lower bound is the sum of the communication and entanglement costs for quantum CDS.
In contrast, our result is only for the shared-randomness cost of a classical CDS protocol, and allows constant error in both correctness and privacy, arbitrary private randomness, and messages of unrestricted length.

\paragraph*{\(f\)-routing.}
Robust lower bounds were previously obtained for the size of a purified adversarial system and for the number of quantum gates~\cite{bluhm2021position,asadi2024linear}, but these measures do not isolate the initially shared entanglement. Asadi, Culf, and May~\cite{asadi2025rank} obtained the first growing lower bounds on the amount of the shared state, under the assumption of perfect recovery on one input class. 

\subsection{Independent work}\label{subsec:concurrent}

In recent independent work, Bogner~\cite{bogner2026robust} proved a logarithmic shared-entanglement lower bound for \(f\)-routing with the inner product function in the two-sided bounded-error regime. The result implies $\mathrm{CDS}(\IP)=\Omega(\log n)$ and other functions that have high sign-rank.

Our robust CDS lower bound is characterized by the logarithm of the deterministic SMP communication complexity. Bogner's result does not subsume our result. In particular, our theorem applies to the equality function, whose sign rank is constant (see, e.g., Remark 1.12 in \cite{hatami2023borsuk}).

We also showed $\mathrm{FR}_0(\IP)=\mathrm{FR}_1(\IP) = \Omega(n)$ where the one-sided perfect assumption is needed.

\subsection{Organization}

In \cref{sec:preliminaries}, we introduce the notation and definitions used throughout the paper.
In \cref{sec:CDSlb}, we prove our shared-randomness lower bound for robust CDS.
In \cref{sec:routing}, we establish our sign-rank lower bounds for one-sided-perfect $f$-routing and derive their consequences for $f$-BB84 and CDS.
In \cref{sec:concluding}, we conclude our work discussing open problems and future directions.
\section{Preliminaries}
\label{sec:preliminaries}
In this section, we introduce some notation and preliminary results useful for our analysis. Throughout the rest of the paper, all logarithms are base-$2$. For a positive integer $q$, we write $[q]:=\{1,\ldots,q\}$. Let $\Delta(\Omega)$ denote the set of probability distributions on a finite set $\Omega$.

\subsection{Basic notation and tools from probability theory}
\label{subsec:basic-tools}

For a vector $z\in\mathbb{R}^q$ and a function
$g:\Omega\to\mathbb{R}$, define
\[
    \|z\|_1:=\sum_{r=1}^q |z_r|,
    \qquad
    \|z\|_\infty:=\max_{r\in[q]}|z_r|,
\]
and
\[
    \|g\|_1:=\sum_{\omega\in\Omega}|g(\omega)|,
    \qquad
    \|g\|_\infty:=\max_{\omega\in\Omega}|g(\omega)|.
\]
If the domain of $g$ is not finite, we use $\|g\|_\infty:=\sup_\omega|g(\omega)|$.
We will use Hölder's inequality in the form
\[
    |\langle u,v\rangle|
    \le \|u\|_1\|v\|_\infty
    \qquad (u,v\in\mathbb R^q).
\]
For random vectors, this inequality can be applied pointwise
and then averaged.

For $P,Q\in\Delta(\Omega)$, their total variation distance is
\[
    \operatorname{TV}(P,Q)
    :=
    \frac12\|P-Q\|_1
    =
    \frac12\sum_{\omega\in\Omega}|P(\omega)-Q(\omega)|.
\]

We use the following operational interpretation of total variation distance.

\begin{fact}[Binary hypothesis testing]
\label{fact:binary-hypothesis-testing}
Let $P,Q\in\Delta(\Omega)$. Suppose that one of $P$ and $Q$ is chosen uniformly at random and a sample is drawn from the chosen distribution. The optimal success and error probabilities for identifying the chosen distribution are
\[
    P_{\rm succ}^\ast
    =
    \frac{1+\operatorname{TV}(P,Q)}{2},
    \qquad
    P_{\rm err}^\ast
    =
    \frac{1-\operatorname{TV}(P,Q)}{2}.
\]
\end{fact}

\begin{proof}
For each observation $\omega$, the optimal rule chooses the distribution assigning the larger probability to $\omega$. Thus
\[
    P_{\rm succ}^\ast
    =
    \frac12\sum_{\omega\in\Omega}\max\{P(\omega),Q(\omega)\}.
\]
The claim follows from $\max\{a,b\}=(a+b+|a-b|)/2$.
\end{proof}

We will also use the following concentration bound.

\begin{fact}[Hoeffding's inequality]
\label{fact:hoeffding}
Let $X_1,\ldots,X_m$ be independent random variables taking values in a common interval $[a,b]$, and let
\[
    \overline X
    :=
    \frac1m\sum_{j=1}^m X_j.
\]
Then, for every $t>0$,
\[
    \Pr\left[
        \left|
            \overline X-\mathbb E[\overline X]
        \right|
        \ge t
    \right]
    \le
    2\exp\left(
        -\frac{2mt^2}{(b-a)^2}
    \right).
\]
\end{fact}

We recall Bayes' rule and its conditional form (see, e.g.,~\cite[Chapter~3]{ross2019first}).
For events $E$ and $F$ with $\Pr[E],\Pr[F]>0$,
\[
    \Pr[E\mid F]
    =
    \frac{\Pr[F\mid E]\Pr[E]}{\Pr[F]}.
\]
More generally, for an additional event $C$ with $\Pr[E\cap C],\Pr[F\cap C]>0$,
\[
    \Pr[E\mid F,C]
    =
    \frac{\Pr[F\mid E,C]\Pr[E\mid C]}
         {\Pr[F\mid C]},
\]
where conditioning on multiple events means conditioning on their intersection. The latter identity follows by applying Bayes' rule to the conditional probability measure $\Pr[\cdot\mid C]$.

\subsection{Deterministic one-way and SMP communication complexity}
\label{subsec:one-way-complexity}

Let \(f:X\times Y\to\{0,1\}\) be a total function. Define its numbers of distinct rows and columns by
\[
    N_{\rm row}(f)
    :=
    \bigl|\{f(x,\cdot):x\in X\}\bigr|,
    \qquad
    N_{\rm col}(f)
    :=
    \bigl|\{f(\cdot,y):y\in Y\}\bigr|.
\]
We denote by \(\D^{\rightarrow}(f)\) and \(\D^{\leftarrow}(f)\) the deterministic one-way communication complexities in the two directions. We also let \(\D^{\|}(f)\) denote the deterministic SMP communication complexity, measured as the total number of bits sent by Alice and Bob. We will use the following standard relations. 

\begin{fact}
\label{fact:deterministic-one-way}
For every total Boolean function \(f\),
\[
    \D^{\rightarrow}(f)
    =
    \left\lceil\log_2 N_{\rm row}(f)\right\rceil,
    \qquad
    \D^{\leftarrow}(f)
    =
    \left\lceil\log_2 N_{\rm col}(f)\right\rceil,
\]
and
\[
    \D^{\|}(f)
    =
    \D^{\rightarrow}(f)+\D^{\leftarrow}(f).
\]

\end{fact}

\begin{proof}
In a deterministic one-way protocol from Alice to Bob, inputs defining distinct rows must produce distinct messages. Conversely, Alice can send the index of the row defined by her input. This proves the first equality, and the column equality follows symmetrically. 

In a deterministic SMP communication protocol, Alice's message must distinguish all distinct rows and Bob's message must distinguish all distinct columns. The resulting lower bound is achieved by having them send their row and column indices, respectively.
\end{proof}

\subsection{Conditional disclosure of secrets (CDS)}
\label{subsec:cds-model}

We use a CDS model in which only Alice receives the secret. The referee knows the input pair $(x,y)$.

\begin{definition}
\label{def:robust-cds}
Let $f:X\times Y\to\{0,1\}$. A one-bit-secret CDS protocol consists of a shared random variable $R$ with $\operatorname{supp}(R)=[q]$, together with independent private randomness for Alice and Bob. Alice receives $(x,s)\in X\times S$ where $S=\{0,1\}$, Bob receives $y\in Y$, and both parties receive the shared value $r\in[q]$. They send messages in finite alphabets $\mathcal M_A$ and $\mathcal M_B$, respectively.

After averaging over the private randomness, Alice's and Bob's encoders are described by conditional output distributions
\[
    P^x_{s,r}\in\Delta(\mathcal M_A),
    \qquad
    Q^y_r\in\Delta(\mathcal M_B),
\]
where
\[
    P^x_{s,r}(a)
    =
    \Pr[A=a\mid X=x,S=s,R=r],
\qquad
    Q^y_r(b)
    =
    \Pr[B=b\mid Y=y,R=r].
\]
Writing $\pi_r:=\Pr[R=r]$, the transcript distribution for secret $s$ and inputs $(x,y)$ is
\begin{equation}
\label{eq:transcript_distribution}
    T_s^{x,y}(a,b)
    :=
    \sum_{r=1}^q \pi_rP^x_{s,r}(a)Q^y_r(b).
\end{equation}

The protocol is \emph{$\varepsilon$-correct} if, for every $(x,y)$ such that $f(x,y)=1$, there exists a decoder
\[
    \mathsf{Dec}_{x,y}:\mathcal M_A\times\mathcal M_B\to\{0,1\}
\]
such that, for each $s\in\{0,1\}$,
\[
    \Pr_{(A,B)\sim T_s^{x,y}}
    \bigl[\mathsf{Dec}_{x,y}(A,B)=s\bigr]
    \ge 1-\varepsilon.
\]
It is \emph{$\delta$-private} if, for every $(x,y)$ such that $f(x,y)=0$, there exists a distribution $\mathsf{Sim}_{x,y}\in\Delta(\mathcal M_A\times\mathcal M_B)$, independent of $s$, such that
\[
    \bigl\|T_s^{x,y}-\mathsf{Sim}_{x,y}\bigr\|_1
    \le\delta
\]
for both $s\in\{0,1\}$. We define the shared-randomness cost as
\[
    \log_2|\operatorname{supp}(R)|=\log_2q.
\]
\end{definition}

\subsection{Sign rank, unbounded-error communication complexity, and the $\IP$ function}
\label{subsec:rank-measures}

\paragraph*{Sign rank.} We use the standard subadditivity of matrix rank,
\begin{equation}
\label{eq:subaddictivity}
    \rank(A+B)\le\rank(A)+\rank(B).
\end{equation}
For a nonzero real number $z$, let $\operatorname{sign}(z)\in\{-1,+1\}$ denote its sign. For a real matrix with no zero entries, $\operatorname{sign}(A)$ is obtained by applying $\operatorname{sign}$ entrywise.

The sign matrix of a Boolean function $f$ is
\[
    S_f(x,y):=2f(x,y)-1\in\{-1,+1\}.
\]
For a sign matrix $S\in\{\pm1\}^{X\times Y}$, its \emph{sign rank} is
\begin{equation}
\label{eq:signrank}
    \signrank(S)
    :=
    \min\left\{
        \rank(A):
        A\in\mathbb R^{X\times Y},\ 
        \operatorname{sign}(A_{x,y})=S_{x,y}
        \text{ for all }x,y
    \right\}.
\end{equation}

\paragraph*{Unbounded-error communication.}
An unbounded-error communication protocol for \(f\) is a private-coin randomized protocol that outputs \(f(x,y)\) with probability strictly greater than \(1/2\) on every input. 
No prescribed constant lower bound on the advantage over \(1/2\) is required.\footnote{In contrast, \emph{bounded-error communication} conventionally requires a success probability of at least \(2/3\) on every input.} 
We denote by \(\UPPcc(f)\) the minimum worst-case communication cost over all such protocols. Paturi and Simon showed the following.

\begin{lemma}[\cite{paturi1986probabilistic}]
\label{lem:upp-signrank}
There exists a constant \(c\ge 0\) such that, for every Boolean function \(f\),
\begin{equation*}
    \log_2\signrank(S_f)
    \ge
    \UPPcc(f)-c.
\end{equation*}
\end{lemma}

\paragraph*{The inner-product function}

For $x,y\in\{0,1\}^n$, define
\[
    \IP_n(x,y):=\sum_{i=1}^n x_i y_i \pmod 2.
\]

Let $H_n$ be the $2^n\times2^n$ Walsh--Hadamard sign matrix
\[
    H_n(x,y):=(-1)^{x\cdot y}.
\]
Under the convention $S_f=2f-1$, we have $S_{\IP_n}=-H_n$. Multiplying a matrix by $-1$ does not change its sign rank. Since $H_n$ is a Hadamard matrix, Forster's spectral lower bound~\cite{forster2002linear} gives the following.

\begin{lemma}[\cite{forster2002linear}]\label{lem:forster-ip}
$\signrank(S_{\IP_n})=\signrank(H_n)\ge2^{n/2}$.
\end{lemma}

\subsection{$f$-routing}
\label{subsec:f-routing}

For density operators $\rho$ and $\sigma$, we define their fidelity by
\[
    F(\rho,\sigma)
    :=
    \operatorname{Tr}\sqrt{\sqrt{\sigma}\rho\sqrt{\sigma}}.
\]
Let \(Q\) be initially maximally entangled with a reference qubit \(\overline Q\), such that \(Q\overline Q\) are in the Bell state
\[
    \ket{\Psi^+}_{Q\overline Q}
    :=
    \frac1{\sqrt2}
    \left(
        \ket{0}_Q\ket{0}_{\overline Q}
        +
        \ket{1}_Q\ket{1}_{\overline Q}
    \right),
    \qquad
    \Psi^+_{Q\overline Q}
    :=
    \ket{\Psi^+}\!\bra{\Psi^+}_{Q\overline Q}.
\]
The reference system \(\overline Q\) remains untouched throughout the protocol. Consequently, all channels acting on \(Q\) are implicitly tensored with the identity channel on $\overline{Q}$.

We use the following $f$-routing model of Asadi, Culf, and May~\cite[Definition~3.1]{asadi2025rank}.

\begin{definition}[$f$-routing]
\label{def:f-routing}
Let $f:X\times Y\to\{0,1\}$. Alice and Bob initially share a bipartite pure state $\ket{\psi}$. Alice receives $x\in X$ and a qubit $Q$, while Bob receives $y\in Y$. Alice applies a channel depending only on $x$ to $Q$ and her share of $\ket{\psi}$, and Bob applies a channel depending only on $y$ to his share of $\ket{\psi}$. They then perform one simultaneous round of quantum communication. Let $M_A$ and $M_B$ be the systems held by Alice and Bob, respectively, after this communication round, and let
\[
    \mathcal N^{x,y}_{Q\to M_AM_B}
\]
be the induced channel before final decoding.

The protocol is $(\varepsilon_0,\varepsilon_1)$-correct if the following hold for every $(x,y)$:
\begin{itemize}
    \item If $f(x,y)=0$, there exists a decoding channel $\mathcal D_A^{x,y}:M_A\to Q$ such that
    \[
        F\left(
            \bigl(\mathcal D_A^{x,y}\circ\operatorname{Tr}_{M_B}
            \circ\mathcal N^{x,y}\bigr)(\Psi^+_{Q\overline Q}),
            \Psi^+_{Q\overline Q}
        \right)
        \ge 1-\varepsilon_0.
    \]

    \item If $f(x,y)=1$, there exists a decoding channel $\mathcal D_B^{x,y}:M_B\to Q$ such that
    \[
        F\left(
            \bigl(\mathcal D_B^{x,y}\circ\operatorname{Tr}_{M_A}
            \circ\mathcal N^{x,y}\bigr)(\Psi^+_{Q\overline Q}),
            \Psi^+_{Q\overline Q}
        \right)
        \ge 1-\varepsilon_1.
    \]
\end{itemize}
\end{definition}

Let $E_{\varepsilon_0,\varepsilon_1}(f)$ be the minimum of
\[
    \log_2\operatorname{SR}(\ket{\psi})
\]
over all $(\varepsilon_0,\varepsilon_1)$-correct $f$-routing protocols using a shared initial pure state $\ket{\psi}$, where $\operatorname{SR}(\ket{\psi})$ denotes its Schmidt rank across the Alice-Bob bipartition. Following~\cite{asadi2025rank}, define
\[
    \FRzero(f):=E_{0,0.05}(f),
    \qquad
    \FRone(f):=E_{0.05,0}(f),
    \qquad
    \pFR(f):=E_{0,0}(f).
\]

\section{Lower bound for shared randomness in robust CDS}\label{sec:CDSlb}

In this section, we show a lower bound for the randomness cost of CDS in the robust setting. 

In \cref{subsec:transcript-gap}, we define the distance \(D(x,y)\) between the transcript distributions corresponding to secrets \(0\) and \(1\). Correctness and privacy force this distance to be large when \(f(x,y)=1\) and small when \(f(x,y)=0\). It follows that distinct rows or columns of \(f\) produce distance functions separated by at least \(\gamma\) $=1-2\varepsilon-\delta$. 
In \cref{subsec:low-dimensional-profiles}, we represent these row and column distance functions using a function class \(\mathcal F_q\). In \cref{subsec:covering-normal-form}, we present a net argument for \(\mathcal F_q\). Intuitively, when \(q\) is small, this class cannot produce too many distinguishable distance functions.
Finally, in \cref{subsec:cds-randomness-lower-bound}, we compare the size of the net with the separation between distance functions.

\subsection{Distance between transcript distributions}\label{subsec:transcript-gap}

Fix an $\varepsilon$-correct and $\delta$-private CDS protocol for $f:X\times Y\to\{0,1\}$, as in \cref{def:robust-cds}. For each input pair $(x,y)$, the protocol induces two transcript distributions, $T^{x,y}_0$ and $T^{x,y}_1$, corresponding to the two possible values of the secret. Define
\[
    {D}(x,y)
    :=
    \operatorname{TV}\bigl(T^{x,y}_0,T^{x,y}_1\bigr).
\]

Correctness forces these distributions to be far apart on inputs of $f$ resulting in $1$, whereas privacy forces them to be close on inputs resulting in $0$.

Set
\[
    \gamma
    :=
    1-2\varepsilon-\delta,
\]
and assume that $\gamma>0$.

\begin{lemma}\label{lem:robust-transcript-gap}
For every $(x,y)\in X\times Y$,
\[
    f(x,y)=0
    \quad\Longrightarrow\quad
    {D}(x,y)\le \delta,
\]
whereas
\[
    f(x,y)=1
    \quad\Longrightarrow\quad
    {D}(x,y)\ge 1-2\varepsilon.
\]
\end{lemma}

\begin{proof}
Suppose first that $f(x,y)=0$. By $\delta$-privacy, there is a distribution $\mathsf{Sim}_{x,y}$, independent of $s$, such that
\[
    \bigl\|T^{x,y}_s-\mathsf{Sim}_{x,y}\bigr\|_1
    \le \delta
\]
for both $s\in\{0,1\}$. The triangle inequality gives
$$
    {D}(x,y)
    =
    \frac{1}{2}
    \bigl\|T^{x,y}_0-T^{x,y}_1\bigr\|_1 
    \le
    \frac{1}{2}
    \left(
        \bigl\|T^{x,y}_0-\mathsf{Sim}_{x,y}\bigr\|_1
        +
        \bigl\|\mathsf{Sim}_{x,y}-T^{x,y}_1\bigr\|_1
    \right) 
    \le \delta.
$$

Now suppose that $f(x,y)=1$. Choose $s$ uniformly from $\{0,1\}$ and apply the CDS decoder to the resulting transcript.
By $\varepsilon$-correctness, its average error probability is at most $\varepsilon$. On the other hand, by \cref{fact:binary-hypothesis-testing}, the minimum possible error in distinguishing $T^{x,y}_0$ from $T^{x,y}_1$ is
\[
    \frac{1-{D}(x,y)}{2}.
\]
Consequently,
\[
    \frac{1-{D}(x,y)}{2}
    \le \varepsilon,
\]
which gives
\[
    {D}(x,y)\ge 1-2\varepsilon.
\]
\end{proof}

We next regard the rows and columns of $D$ as functions. For each $x\in X$, define
\[
    D_x:Y\to[0,1],
    \qquad
    D_x(y):=D(x,y),
\]
and, for each $y\in Y$, define
\[
    D^{y}:X\to[0,1],
    \qquad
    D^{y}(x):=D(x,y).
\]

\begin{lemma}\label{lem:profile-separation}
If $x,x'\in X$ determine distinct rows of $f$, then
\[
    \bigl\|{D}_x-{D}_{x'}\bigr\|_\infty
    \ge \gamma.
\]
Similarly, if $y,y'\in Y$ determine distinct columns of $f$, then
\[
    \bigl\|{D}^{\,y}-{D}^{\,y'}\bigr\|_\infty
    \ge \gamma.
\]
\end{lemma}

\begin{proof}
Suppose that $f(x,\cdot)\neq f(x',\cdot)$. Then there is some $y_0\in Y$ such that
\[
    f(x,y_0)\neq f(x',y_0).
\]
One of these two values is $0$, and the other is $1$. By \cref{lem:robust-transcript-gap}, one of ${D}(x,y_0)$ and ${D}(x',y_0)$ is at most $\delta$, whereas the other is at least $1-2\varepsilon$. Hence
\[
    \left|
        {D}(x,y_0)-{D}(x',y_0)
    \right|
    \ge
    1-2\varepsilon-\delta
    =
    \gamma.
\]
Taking the maximum over $y\in Y$ yields
\[
    \bigl\|{D}_x-{D}_{x'}\bigr\|_\infty
    \ge \gamma.
\]

The column statement follows identically. If $f(\cdot,y)\neq f(\cdot,y')$, choose $x_0\in X$ such that $f(x_0,y)\neq f(x_0,y')$ and apply the same reasoning.
\end{proof}

Thus the distinct rows of the communication matrix of $f$ give a pairwise $\gamma$-separated family of row-distance functions $\{{D}_x:x\in X\}$. Likewise, the distinct columns give a pairwise $\gamma$-separated family of column-distance functions $\{{D}^{\,y}:y\in Y\}$. To bound the sizes of these families, we next express both kinds of distance functions using a class of functions \(\mathcal F_q\).

\subsection{Low-dimensional representations of row and column distance functions}
\label{subsec:low-dimensional-profiles}
We now show that the row and column distance functions defined in \cref{subsec:transcript-gap} are generated by a class of functions whose domain has dimension $q$, the support size of the shared random variable.

Let
\[
    B_1^q
    :=
    \left\{
        u\in\mathbb{R}^q:\|u\|_1\le 1
    \right\}
\]
denote the set of vectors in $\mathbb{R}^q$ whose $\ell_1$ norm is at most $1$, and let
\[
    \Delta_q
    :=
    \left\{
        p\in\mathbb{R}_{\ge 0}^q:
        \sum_{r=1}^q p_r=1
    \right\}
    \subseteq B_1^q
\]
be the set of probability distributions on $[q]$.

Define \(\mathcal F_q\) to be the class of functions \(g:B_1^q\to[0,1]\) of the form
\begin{equation}\label{eq:normal-form-class}
    g(u)
    =
    \frac{1}{2}
    \sum_{i\in I}
    \bigl|\langle u,z_i\rangle\bigr|,
\end{equation}
where \(I\) is an arbitrary finite index set, the vectors \(z_i\in\mathbb{R}^q\) satisfy
\begin{equation}\label{eq:normal-form-column-bound}
    \sum_{i\in I}|z_i(r)|
    \le 2
    \qquad
    \text{for every }r\in[q],
\end{equation}
and \(\langle\cdot,\cdot\rangle\) denotes the standard inner product on \(\mathbb{R}^q\).

The bound in \cref{eq:normal-form-column-bound} ensures that these functions take values in \([0,1]\). Indeed, for every \(u\in B_1^q\),
\[
    g(u) \le \frac{1}{2} \sum_{i\in I} \sum_{r=1}^q |u_r|\,|z_i(r)| 
    =
    \frac{1}{2} \sum_{r=1}^q |u_r| \sum_{i\in I}|z_i(r)| 
    \le
    \sum_{r=1}^q|u_r|
    \le 1.
\]

\paragraph{Posterior representation for row-distance functions.}

We next give two representations of the transcript distance \(D(x,y)\). The first controls its row-distance functions, while the second controls its column-distance functions.

For each \(x\in X\) and \(a\in\mathcal M_A\), define a vector
\(z^x_a\in\mathbb{R}^q\) by
\[
    z^x_a(r)
    :=
    P^x_{0,r}(a)-P^x_{1,r}(a).
\]
For every \(r\in[q]\),
\[
    \sum_{a\in\mathcal M_A}|z^x_a(r)|
    \le
    \sum_{a\in\mathcal M_A}
    \left(
        P^x_{0,r}(a)+P^x_{1,r}(a)
    \right) =2.
\]
Consequently, the function
\[
    \phi_x(u)
    :=
    \frac{1}{2}
    \sum_{a\in\mathcal M_A}
    \bigl|
        \langle u,z^x_a\rangle
    \bigr|
\]
belongs to \(\mathcal F_q\).

For fixed \(y\in Y\), define the distribution of Bob's message on \(\mathcal M_B\), averaged over the shared randomness, by
\[
    \nu_y(b) :=
    \sum_{r=1}^q
    \pi_r Q^y_r(b) = \mathrm{Pr}[B=b|Y=y],
\]
for $b\in\mathcal M_B$.
Whenever \(\nu_y(b)>0\), define
\[
    p_{y,b}(r)
    :=
    \frac{\pi_r Q^y_r(b)}{\nu_y(b)} = \mathrm{Pr}[R=r|Y=y,B=b],
\]
for $r\in[q]$. The second equality follows from Bayes' rule, since the shared randomness $R$ is independent of the input $Y$.
The vector \(p_{y,b}\) belongs to \(\Delta_q\); it is the conditional distribution of \(R\) given \(Y=y\) and \(B=b\).
When \(\nu_y(b)=0\), define \(p_{y,b}\) arbitrarily in \(\Delta_q\).

Let \(\mu_y\) be the distribution of \(p_{y,b}\) obtained by sampling \(b\) according to \(\nu_y\). 

\begin{lemma}
\label{lem:posterior-representation}
For every \(x\in X\), the function \(\phi_x\) belongs to \(\mathcal F_q\), and the row-distance function of \(x\) satisfies
\[
     D_x(y)
    =
    \mathbb E_{p\sim\mu_y}
    \bigl[\phi_x(p)\bigr]
    \qquad\text{for every }y\in Y.
\]
\end{lemma}

\begin{proof}
Using the definition of the transcript distributions in \cref{eq:transcript_distribution}, we have
$$
    D_x(y)
    =
    \frac{1}{2}
    \sum_{a\in\mathcal M_A}
    \sum_{b\in\mathcal M_B}
    \left|
        T^{x,y}_0(a,b)-T^{x,y}_1(a,b)
    \right| 
    =
    \frac{1}{2}
    \sum_{a\in\mathcal M_A}
    \sum_{b\in\mathcal M_B}
    \left|
        \sum_{r=1}^q
        \pi_r z^x_a(r)Q^y_r(b)
    \right|.
$$
For every \(b\) with \(\nu_y(b)>0\),
\[
    \pi_r Q^y_r(b)
    =
    \nu_y(b)p_{y,b}(r).
\]
Messages satisfying \(\nu_y(b)=0\) contribute nothing to the sum.
Therefore,
$$
    D_x(y) 
    =
    D(x,y)
    =
    \sum_{b\in\mathcal M_B}
    \nu_y(b)
    \left(
        \frac{1}{2}
        \sum_{a\in\mathcal M_A}
        \left|
            \sum_{r=1}^q
            p_{y,b}(r)z^x_a(r)
        \right|
    \right) 
    =
    \sum_{b\in\mathcal M_B}
    \nu_y(b)\phi_x(p_{y,b}) 
    =
    \mathbb E_{p\sim\mu_y}
    \bigl[\phi_x(p)\bigr].
$$
\end{proof}

Thus, once we fix a CDS protocol, each \(x\) determines a function \(\phi_x\in\mathcal F_q\), while each \(y\) determines a distribution \(\mu_y\). The entire row-distance function \( D_x\) is obtained by evaluating \(\phi_x\) against these distributions. 

\paragraph{Dual representation for column-distance functions.}

The preceding representation is obtained by conditioning on Bob's message.
To obtain an analogous representation for columns, we instead condition on Alice's message. Since Alice’s message depends on the secret, we keep track of the difference between the posterior probabilities corresponding to the two secret values. 

For this construction, let the secret $S$ be uniform on $\{0,1\}$ and independent of the shared randomness and the inputs. For each \(x\in X\), define the marginal distribution of Alice's message by
\[
    \nu_x(a)
    :=
    \frac{1}{2}
    \sum_{r=1}^q
    \pi_r
    \left(
        P^x_{0,r}(a)+P^x_{1,r}(a)
    \right)
    =
    \Pr[A=a\mid X=x].
\]
Whenever \(\nu_x(a)>0\), define \(w_{x,a}\in\mathbb R^q\) by
\[
    w_{x,a}(r)
    :=
    \frac{\pi_r\bigl(P^x_{0,r}(a)-P^x_{1,r}(a)\bigr)}
         {2\nu_x(a)}
    =
    \Pr[R=r,S=0\mid X=x,A=a]
    -
    \Pr[R=r,S=1\mid X=x,A=a].
\]
The second equality follows from Bayes' rule, since $\Pr[R=r,S=s\mid X=x]=\pi_r/2$ and $P_{s,r}^x(a)=\Pr[A=a\mid X=x,R=r,S=s]$. Thus, for each \(r\), \(w_{x,a}(r)\) records the difference between the posterior probabilities of \((R,S)=(r,0)\) and \((R,S)=(r,1)\) after observing Alice’s message.
If \(\nu_x(a)=0\), set \(w_{x,a}=0\).

For every \(a\) with \(\nu_x(a)>0\),
\[
\begin{aligned}
    \|w_{x,a}\|_1
    &\le
    \sum_{r=1}^q
    \Bigl(
        \Pr[R=r,S=0\mid X=x,A=a]
        +
        \Pr[R=r,S=1\mid X=x,A=a]
    \Bigr)\\
    &=1.
\end{aligned}
\]
Hence \(w_{x,a}\in B_1^q\). Let \(\lambda_x\) be the distribution of \(w_{x,a}\) obtained by sampling \(a\) according to \(\nu_x\). 

We now associate a function with each \(y\in Y\). Define, for $w\in B_1^q$, 
\[
    \psi_y(w)
    :=
    \sum_{b\in\mathcal M_B}
    \left|
        \sum_{r=1}^q
        w(r) Q^y_r(b)
    \right|.
\]
Recall that \(Q^y_r\) is a probability distribution on Bob's message $b$, induced by the shared randomness $r$ and the input $y$. 

Moreover, $\psi_y$ belongs to $\mathcal F_q$, since
\[
    \psi_y(w)
    =
    \frac12
    \sum_{b\in\mathcal M_B}
    \left|
        \sum_{r=1}^q w(r)\,2Q_r^y(b)
    \right|,
\]
and, for every $r\in[q]$,
\[
    \sum_{b\in\mathcal M_B}|2Q_r^y(b)|
    =
    2\sum_{b\in\mathcal M_B}Q_r^y(b)
    =
    2.
\]

\begin{lemma}
\label{lem:dual-representation}
For every \(y\in Y\), the function \(\psi_y\) belongs to \(\mathcal F_q\), and the column-distance function of \(y\) satisfies
\[
    D^{\,y}(x)
    =
    \mathbb E_{w\sim\lambda_x}
    \bigl[\psi_y(w)\bigr]
    \qquad\text{for every }x\in X.
\]
\end{lemma}

\begin{proof}
If \(\nu_x(a)=0\), then
\[
    \pi_r\bigl(P^x_{0,r}(a)+P^x_{1,r}(a)\bigr)=0
\]
for every \(r\), and consequently \(\pi_r z^x_a(r)=0\). Therefore,
\[
\begin{aligned}
    \mathbb E_{w\sim\lambda_x}
    \bigl[\psi_y(w)\bigr]
    &=
    \sum_{\substack{a\in\mathcal M_A\\ \nu_x(a)>0}}
    \nu_x(a)
    \sum_{b\in\mathcal M_B}
    \left|
        \sum_{r=1}^q
        \frac{\pi_r z^x_a(r)}
             {2\nu_x(a)}
        Q^y_r(b)
    \right|\\
    &=
    \frac{1}{2}
    \sum_{a\in\mathcal M_A}
    \sum_{b\in\mathcal M_B}
    \left|
        \sum_{r=1}^q
        \pi_r z^x_a(r)Q^y_r(b)
    \right|\\
    &=
    D(x,y)
    =
    D^{\,y}(x).
\end{aligned}
\]
\end{proof}

\subsection{Covering the function class ${\mathcal F_q}$}
\label{subsec:covering-normal-form}
We now construct a finite family that approximates every function in \(\mathcal F_q\) in the $\infty$-norm. For \(0<\eta\le 1\), let
\[
    \mathcal N_\infty(\eta,\mathcal F_q)
\]
denote the smallest cardinality of a finite family \(\mathcal H\) of functions on \(B_1^q\) such that, for every \(g\in\mathcal F_q\), there exists \(h\in\mathcal H\) satisfying
\[
    \|g-h\|_\infty
    \le \eta.
\]
The functions in \(\mathcal H\) are not required to belong to \(\mathcal F_q\).

\begin{lemma}
\label{lem:normal-form-cover}
There is a constant \(c>0\) such that, for every \(q\ge 1\) and \(0<\eta\le 1\),
\[
    \log_2\mathcal N_\infty(\eta,\mathcal F_q)
    \le
    c\frac{q^4}{\eta^2}
    \log_2^2\left(\frac{cq}{\eta}\right).
\]
\end{lemma}

\begin{proof}
Fix an integer $q\ge 1$ and $0<\eta\le 1$.
Let
\[
    \Gamma
    :=
    \left(
        \frac{\eta}{4q}\mathbb Z
    \right)\cap[-1,1].
\]
Every point of $[-1,1]$ can be rounded to a point of $\Gamma$ with error at most $\eta/(4q)$, and
\[
    |\Gamma|
    \le
    1+\frac{8q}{\eta}
    \le
    \frac{10q}{\eta}.
\]
Let $\mathcal U\subseteq B_1^q$ be an $\eta/(8q)$-net in the $\ell_1$-norm.
By the standard volumetric bound (see, e.g., Proposition C.3 in \cite{foucart2013mathematical}), we may choose
\[
    |\mathcal U|
    \le
    \left(1+\frac{16q}{\eta}\right)^q
    \le
    \left(\frac{17q}{\eta}\right)^q.
\]
Set
\[
    M
    :=
    \left\lceil
        \frac{8q^2}{\eta^2}
        \ln\bigl(4|\mathcal U|\bigr)
    \right\rceil.
\]

For each $\mathbf w=(w_1,\ldots,w_M)\in(\Gamma^q)^M$, define $h_{\mathbf w}:B_1^q\to\mathbb R$ by
\[
    h_{\mathbf w}(u)
    :=
    \frac qM
    \sum_{j=1}^M|\langle u,w_j\rangle|.
\]
Consider the finite family
\[
    \mathcal H
    :=
    \left\{
        h_{\mathbf w}:
        \mathbf w\in(\Gamma^q)^M
    \right\}.
\]
This family depends only on $q$ and $\eta$, and
\[
    |\mathcal H|
    \le
    |\Gamma|^{qM}
    \le
    \left(\frac{10q}{\eta}\right)^{qM}.
\]
We show that \(\mathcal H\) is an \(\eta\)-cover of \(\mathcal F_q\) with respect to the $\infty$ norm
\[
    \|g-h\|_\infty
    :=
    \sup_{u\in B_1^q}|g(u)-h(u)|.
\]

Fix $g\in\mathcal F_q$, and write
\[
    g(u)
    =
    \frac12\sum_{i\in I}|\langle u,z_i\rangle|,
    \qquad
    \sum_{i\in I}|z_i(r)|\le 2
    \quad\text{for every }r\in[q].
\]
Then
\[
    \sum_{i\in I}\|z_i\|_\infty
    \le
    \sum_{r=1}^q\sum_{i\in I}|z_i(r)|
    \le
    2q.
\]

Define $V$ by selecting a nonzero term $z_i$ with probability $\|z_i\|_\infty/(2q)$ and setting $V=z_i/\|z_i\|_\infty$; with the remaining probability, set $V=0$.
Then, for every $u\in B_1^q$,
\[
\begin{aligned}
    q\,\mathbb E|\langle u,V\rangle|
    &=
    q\sum_{\substack{i\in I\\z_i\ne0}}
    \frac{\|z_i\|_\infty}{2q}
    \left|
        \left\langle u,\frac{z_i}{\|z_i\|_\infty}\right\rangle
    \right|\\
    &=
    \frac12\sum_{i\in I}|\langle u,z_i\rangle|
    =g(u).
\end{aligned}
\]

Round each coordinate of $V$ to a point of $\Gamma$, obtaining a random vector $W\in\Gamma^q$ such that
\[
    \|V-W\|_\infty
    \le
    \frac{\eta}{4q}.
\]
For every $u\in B_1^q$, it follows that
$$
        \left|
            g(u)-q\,\mathbb E|\langle u,W\rangle|
        \right|
        \le
        q\,\mathbb E|\langle u,V-W\rangle|
        \le
        q\|u\|_1\,\mathbb E\|V-W\|_\infty
        \le
        \frac{\eta}{4}.
$$

Now draw independent copies $W_1,\ldots,W_M$ of $W$ and let
\[
    h(u)
    :=
    \frac qM
    \sum_{j=1}^M|\langle u,W_j\rangle|.
\]
Every realization of $h$ belongs to $\mathcal H$.
For each fixed $u\in B_1^q$, the independent random variables $q|\langle u,W_j\rangle|$ lie in $[0,q]$. Applying Hoeffding's inequality (\cref{fact:hoeffding}) for each $u\in\mathcal U$ and taking a union bound over $\mathcal U$, we obtain
\[
    \Pr\left[
        \max_{u\in\mathcal U}
        \left|
            h(u)-q\,\mathbb E|\langle u,W\rangle|
        \right|
        >
        \frac{\eta}{4}
    \right]
    \le
    2|\mathcal U|
    \exp\left(-\frac{M\eta^2}{8q^2}\right)
    \le
    \frac12.
\]
Consequently, there exists a realization $h\in\mathcal H$ such that
\[
    \max_{u\in\mathcal U}|g(u)-h(u)|
    \le
    \frac{\eta}{2}.
\]

We next observe that both \(g\) and \(h\) are \(q\)-Lipschitz with respect to the \(\ell_1\)-norm. Indeed, for all \(u,v\in B_1^q\),
$$
    |g(u)-g(v)|
    \le
    q\,\mathbb E
    \left|
        |\langle u,V\rangle|
        -
        |\langle v,V\rangle|
    \right|
    \le
    q\,\mathbb E|\langle u-v,V\rangle|
    \le
    q\|u-v\|_1\mathbb E\|V\|_\infty
    \le
    q\|u-v\|_1,
$$
since \(V\in[-1,1]^q\). Similarly,
$$
    |h(u)-h(v)|
    \le
    \frac qM\sum_{j=1}^M
    |\langle u-v,W_j\rangle|
    \le
    \frac qM\sum_{j=1}^M
    \|u-v\|_1\|W_j\|_\infty
    \le
    q\|u-v\|_1,
$$
since \(W_j\in[-1,1]^q\) for every \(j\).

For an arbitrary $u\in B_1^q$, choose $u_0\in\mathcal U$ with $\|u-u_0\|_1\le\eta/(8q)$.
Then
\[
    \begin{aligned}
        |g(u)-h(u)|
        &\le
        |g(u)-g(u_0)|
        +|g(u_0)-h(u_0)|
        +|h(u_0)-h(u)|\\
        &\le
        2q\|u-u_0\|_1+\frac{\eta}{2}\\
        &\le
        \frac{3\eta}{4}
        <\eta.
    \end{aligned}
\]
Thus $\mathcal H$ is an $\eta$-cover of $\mathcal F_q$ in the $\infty$-norm.

Finally, the bounds on $|\mathcal U|$ and $M$ imply
$$
        M
        \le
        \frac{8q^2}{\eta^2}
        \left(
            \ln 4+q\ln\left(\frac{17q}{\eta}\right)
        \right)+1
        \le
        \frac{17q^3}{\eta^2}
        \ln\left(\frac{17q}{\eta}\right).
$$

Therefore,
$$
        \log_2\mathcal N_\infty(\eta,\mathcal F_q)
        \le
        \log_2|\mathcal H|
        \le
        qM\log_2\left(\frac{10q}{\eta}\right)
        \le
        \frac{17q^4}{\eta^2}
        \log_2^2\left(\frac{17q}{\eta}\right).
$$
This proves the claimed bound.
\end{proof}

\subsection{Lower bound for shared randomness}
\label{subsec:cds-randomness-lower-bound}

We now combine the separation of distance functions in \cref{lem:profile-separation} with the covering bound for \(\mathcal F_q\) in~\cref{subsec:covering-normal-form} to prove the main theorem in this section. 

\begin{theorem}
\label{thm:cds-randomness-lower-bound}
Let \(f:X\times Y\to\{0,1\}\) be a nonconstant Boolean function, and suppose that \(f\) admits an \(\varepsilon\)-correct and \(\delta\)-private CDS protocol using a
shared random variable supported on \([q]\). Set 
\[
    \gamma:=1-2\varepsilon-\delta,
\]
and suppose that \(\gamma>0\). Then there is a universal constant
\(C>0\) such that 
\begin{equation}
\label{eq:communication-upper-bound}
    \D^{\|}(f)
    \le
    C\frac{q^4}{\gamma^2}
    \log_2^2\left(\frac{Cq}{\gamma}\right)
    +1.
\end{equation}
Consequently, for every fixed constant \(\gamma>0\),
\[
    \log_2 q
    =
    \Omega\left(\log_2\D^{\|}(f)\right).
\]
\end{theorem}

\begin{proof}
Set
\[
    \eta:=\frac{\gamma}{4},
\]
and let \(\mathcal H\) be an \(\eta\)-net for \(\mathcal F_q\) with
\[
    |\mathcal H|
    =
    \mathcal N_\infty(\eta,\mathcal F_q),
\]
as proved in~\cref{lem:normal-form-cover}.

We first bound the number of distinct rows of \(f\). For \(h\in\mathcal H\), define a function \(R_h:Y\to\mathbb R\) by
\[
    R_h(y)
    :=
    \mathbb E_{p\sim\mu_y}[h(p)].
\]
For every \(x\in X\), choose \(h_x\in\mathcal H\) such that
\[
    \|\phi_x-h_x\|_\infty
    \le \eta.
\]
By \cref{lem:posterior-representation},
\[
    D_x(y)
    =
    \mathbb E_{p\sim\mu_y}[\phi_x(p)].
\]
Hence, for every \(y\in Y\),
\[
    |D_x(y)-R_{h_x}(y)|
    =
    \left|
        \mathbb E_{p\sim\mu_y}
        [\phi_x(p)-h_x(p)]
    \right| 
    \le
    \|\phi_x-h_x\|_\infty 
    \le \eta.
\]
Therefore,
\[
    \|D_x-R_{h_x}\|_\infty
    \le \eta.
\]

According to~\cref{lem:profile-separation}, distinct rows of \(f\) give distance functions separated by at least \(\gamma\). Since
\[
    2\eta=\frac{\gamma}{2}<\gamma,
\]
two distinct rows cannot be approximated by the same function \(R_h\).
It follows that
\[
    N_{\mathrm{row}}(f)
    \le
    |\mathcal H|.
\]

We apply the same argument for columns of $f$. For \(h\in\mathcal H\), define \(C_h:X\to\mathbb R\) by
\[
    C_h(x)
    :=
    \mathbb E_{w\sim\lambda_x}[h(w)].
\]

Using \cref{lem:dual-representation} and~\cref{lem:profile-separation}, we similarly obtain
\[
    N_{\mathrm{col}}(f)
    \le
    |\mathcal H|.
\]

It now follows from \cref{fact:deterministic-one-way} that
$$
    \D^{\|}(f)
    =
    \D^\to(f)+\D^\leftarrow(f)
    \le
    2\log_2|\mathcal H|+2
    =
    2\log_2\mathcal N_\infty(\eta,\mathcal F_q)+2.
$$
Applying \cref{lem:normal-form-cover} and substituting \(\eta=\gamma/4\) gives
$$
    \D^{\|}(f)
    \le
    2c\frac{q^4}{\eta^2}
    \log_2^2\left(\frac{cq}{\eta}\right)+2
    \le
    32c\frac{q^4}{\gamma^2}
    \log_2^2\left(\frac{4cq}{\gamma}\right)+2
    \le
    C\frac{q^4}{\gamma^2}
    \log_2^2\left(\frac{Cq}{\gamma}\right)+1
$$
for a sufficiently large universal constant \(C\). This proves~\cref{eq:communication-upper-bound}.

For every fixed constant \(\gamma>0\), the preceding bound gives
\[
    \D^{\|}(f)=q^{O(1)}.
\]
Taking logarithms therefore yields
\[
    \log_2q
    =
    \Omega\left(\log_2\D^{\|}(f)\right).\qedhere
\]
\end{proof}
\section{Lower bound for $f$-routing}\label{sec:routing}

We use the following rank method of Asadi, Culf, and May \cite{asadi2025rank}.

\begin{proposition}[Adapted from Definition~3.3, Lemma~3.4 and Lemma~3.5 in \cite{asadi2025rank}]\label{prop:acm}
Suppose that an $f$-routing protocol is perfect on the $0$-inputs and has recovery error at most $0.05$ on the $1$-inputs. Suppose further that its shared pure resource state has Schmidt rank $d_E$. Then there exists a real matrix $G\in\mathbb R^{X\times Y}$ such that
\[
G_{x,y}=0 \quad\text{if } f(x,y)=0,
\qquad
G_{x,y}>0 \quad\text{if } f(x,y)=1,
\]
and
\[
\rank(G)\le d_E^4.
\]
The symmetric statement holds for protocols perfect on the $1$-inputs and having error at most $0.05$ on the $0$-inputs.
\end{proposition}

For the joint state $\rho_{\overline Q M_B}^{x,y}$ of the reference system and Bob's system after the communication round and before decoding, Asadi, Culf, and May~\cite{asadi2025rank} defined
\[
    G_{x,y}
    :=
    \operatorname{tr}\left[
        \left(
            \rho_{\overline Q M_B}^{x,y}
            -
            \frac{I_{\overline Q}}{d_{\overline Q}}
            \otimes\rho_{M_B}^{x,y}
        \right)^2
    \right]
\]
where $d_{\overline Q}:=\dim \overline Q$.
This is a squared Hilbert--Schmidt norm, and it is a real nonnegative value and vanishes exactly when \(\overline Q\) and \(M_B\) are uncorrelated. Their decoupling argument shows that this occurs exactly on the \(0\)-inputs under the assumptions of \cref{prop:acm}. Therefore, \(G_{x,y}>0\) exactly when \(f(x,y)=1\). 

\begin{theorem}[Sign-rank lower bounds for $f$-routing]\label{thm:main}
Let \(f:X\times Y\to\{0,1\}\) be a nonconstant Boolean function. Any protocol covered by \cref{prop:acm}, using a shared state of Schmidt rank $d_E$, satisfies
\[
 d_E^4\ge \signrank(S_f)-1.
\]
Consequently,
\[
 \FRzero(f),\ \FRone(f),\ \pFR(f)
 \ge
 \frac14\log \bigl(\max\{1,\signrank(S_f)-1\}\bigr).
\]
\end{theorem}

\begin{proof}
First suppose the protocol is perfect on the $0$-inputs, and let $G$ be the matrix from \cref{prop:acm}. We define
\[
 m:=\min_{f(x,y)=1}G_{x,y}>0.
\]
Let $J$ be the all-ones matrix, choose $0<\delta<m$, and define
\[
 A:=G-\delta J.
\]
If $f(x,y)=0$, then $A_{x,y}=-\delta<0$; if $f(x,y)=1$, then $A_{x,y}\ge m-\delta>0$. Hence
\[
 \operatorname{sign}(A)=S_f.
\]
Therefore
\[
 \signrank(S_f)
 \le \rank(A)
 \le \rank(G)+\rank(J)
 \le d_E^4+1.
\]
The first inequality holds from the definition of the sign rank \cref{eq:signrank}, and the second one follows from the subadditivity of the rank of matrices \cref{eq:subaddictivity}.
This proves $d_E^4\ge \signrank(S_f)-1$.

If the protocol is perfect on the $1$-inputs, apply the same argument to $1-f$. Since $S_{1-f}=-S_f$ and a global sign change does not alter sign rank, the same bound follows.  Exact routing is a special case of both one-sided-perfect settings.  Taking logarithms proves the cost bound.
\end{proof}

Combining the theorem with the standard facts recalled above gives the following operational form and its basic application.

\begin{corollary}\label{cor:upp-ip}
There exists a constant $c'$ such that, for every non-constant Boolean function $f$,
\[
 \FRzero(f),\ \FRone(f),\ \pFR(f)
 \ge \frac14\UPPcc(f)-c'.
\]
In particular,
\[
 \FRzero(\IP_n),\ \FRone(\IP_n),\ \pFR(\IP_n)
 \geq
 \frac{n}{8}-\frac{1}{4}.
\]
\end{corollary}

\begin{proof}
The first statement follows from \cref{thm:main} and~\cref{lem:upp-signrank}.  For inner product, we use the sign-rank lower bound (\cref{lem:forster-ip}) and get
\[
 \frac14\log \bigl(\signrank(S_{\IP_n})-1\bigr)
 \geq \frac{n}{8}-\frac{1}{4}.
\]
\end{proof}

\paragraph*{Lower bounds for \(f\)-BB84.}
In the \(f\)-BB84 task, Alice receives \(x\) and a qubit in the BB84 state \(H^{f(x,y)}\ket{b}\), while Bob receives \(y\); after one simultaneous round of communication, both parties must output \(b\).
A protocol is \((\varepsilon_0,\varepsilon_1)\)-correct if it succeeds with probability at least \(1-\varepsilon_0\) whenever \(f(x,y)=0\), and with probability at least \(1-\varepsilon_1\) whenever \(f(x,y)=1\). 

Let
\[
    \operatorname{FBB84}_0(f),\qquad
    \operatorname{FBB84}_1(f),\qquad
    \operatorname{pFBB84}(f)
\]
denote the minimum logarithmic Schmidt-rank costs of \(f\)-BB84 protocols that are respectively \((0,\varepsilon)\)-correct, \((\varepsilon,0)\)-correct, and \((0,0)\)-correct, where \(\varepsilon>0\) is a sufficiently small constant.

\begin{corollary}
\label{cor:fbb84}
For every Boolean function \(f\),
\[
    \operatorname{FBB84}_0(f),\
    \operatorname{FBB84}_1(f),\
    \operatorname{pFBB84}(f)
    =
    \Omega\left(
        \log_2\bigl(\max\{1,\signrank(S_f)-1\}\bigr)
    \right).
\]
Consequently, there exist constants $c_0,c_1$ such that, 
\[
    \operatorname{FBB84}_0(f),\
    \operatorname{FBB84}_1(f),\
    \operatorname{pFBB84}(f)
    \ge
    c_0\UPPcc(f)-c_1.
\]
\end{corollary}

\begin{proof}
The \(f\)-BB84 task is called the \(f\)-measure task in \cite{bluhm2026complexity}. The reductions proven in \cite[Theorem~1, Section~4.1]{bluhm2026complexity} transform \(f\)-BB84 protocols into \(f\)-routing protocols using only a constant number of copies of the shared resource state. Moreover, the reduction preserves zero error and sufficiently small constant error. Hence,
\[
\begin{aligned}
    \operatorname{FBB84}_0(f)
    &=\Omega\bigl(\FRzero(f)\bigr),\\
    \operatorname{FBB84}_1(f)
    &=\Omega\bigl(\FRone(f)\bigr),\\
    \operatorname{pFBB84}(f)
    &=\Omega\bigl(\pFR(f)\bigr).
\end{aligned}
\]
The sign-rank bounds now follow from \cref{thm:main}. The UPP bounds follow from~\cref{lem:upp-signrank}.
\end{proof}

\paragraph*{Lower bounds for CDS.}

We next transfer our \(f\)-routing lower bounds to classical CDS.

Fix a sufficiently small constant $\eta>0$. Let
\[
    \operatorname{ppCDS}(f),\qquad
    \operatorname{pcCDS}(f),\qquad
    \operatorname{pCDS}(f)
\]
denote the minimum shared-randomness costs of CDS protocols for $f$ that are, respectively, $\eta$-correct and perfectly private, perfectly correct and $\eta$-private, and both perfectly correct and perfectly private.

\begin{corollary}\label{cor:cds}
For every nonconstant Boolean function \(f\),
\[
    \operatorname{ppCDS}(f),\
    \operatorname{pcCDS}(f),\
    \operatorname{pCDS}(f)
    =
    \Omega\left(
        \log_2\bigl(\max\{1,\signrank(S_f)-1\}\bigr)
    \right).
\]
Consequently, there exist constants $c_1,c_0$ such that
\[
    \operatorname{ppCDS}(f),\
    \operatorname{pcCDS}(f),\
    \operatorname{pCDS}(f)
    \ge
    c_0\UPPcc(f)-c_1.
\]
\end{corollary}

\begin{proof}
The CDS-to-CDQS reduction of Allerstorfer et al.~\cite[Theorem~22]{allerstorfer2024relating}, followed by their CDQS-to-routing reduction~\cite[Theorem~23]{allerstorfer2024relating}, preserves the shared-resource cost up to constant factors and therefore gives
\[
\begin{aligned}
    \operatorname{ppCDS}(f)
    &=\Omega\bigl(\FRzero(f)\bigr),\\
    \operatorname{pcCDS}(f)
    &=\Omega\bigl(\FRone(f)\bigr),\\
    \operatorname{pCDS}(f)
    &=\Omega\bigl(\pFR(f)\bigr).
\end{aligned}
\]
The sign-rank bounds now follow from \cref{thm:main}. The UPP bounds follow from~\cref{lem:upp-signrank}.
\end{proof}
\section{Concluding remarks}\label{sec:concluding}

We studied lower bounds on shared randomness for CDS protocols and on shared entanglement for $f$-routing. Our robust CDS lower bound is expressed in terms of the logarithm of exact deterministic SMP complexity and holds even when communication and private randomness are unrestricted. For equality, this bound is tight up to constant factors. For $f$-routing, we used the structure matrix of Asadi, Culf, and May~\cite{asadi2025rank} to obtain a sign-rank lower bound, or equivalently a lower bound in terms of $\mathsf{UPP}$ communication complexity. In particular, this yields an $\Omega(n)$ entanglement lower bound for the inner product function when recovery is perfect on either the $0$-inputs or the $1$-inputs and has constant error on the other input class.

In independent work, Bogner~\cite{bogner2026robust} proved an $\Omega(\log n)$ entanglement lower bound for $f$-routing with the inner product function, allowing constant error on both input classes. Since the deterministic SMP complexity is at most linear in the input length $n$, our robust CDS bound yields at most a logarithmic lower bound on shared-randomness cost. A natural related question is whether CDS for the inner product function requires $n^{\Omega(1)}$ bits of shared randomness when constant correctness and privacy errors are allowed, with unrestricted communication and private randomness.
We believe developing stronger techniques for shared-randomness lower bounds in CDS will provide useful ideas for the quantum setting.

Finally, it would also be interesting to obtain a direct operational understanding of the relationship between $\mathsf{UPP}$ protocols, CDS, and $f$-routing. Our entanglement lower bound relates $f$-routing and CDS to unbounded-error communication through sign rank.
Can this relation be explained directly in terms of protocol operations? A similar question related to $\mathsf{QNP}$ communication protocols \cite{deWolf2003nondeterministic} was asked in \cite{asadi2025rank}.

\section*{AI disclosure}

The authors first observed that if a CDS protocol for $f$ uses no randomness, then $f(x,y)=g(x)$ holds for some function $g$. They developed its generalization for which $f$ there is a CDS protocol with small randomness using ChatGPT 5.6 Sol. They also found that the authors in \cite{asadi2025rank} did not obtain a lower bound for the inner product function using their rank method. ChatGPT 5.6 Sol then found that the tight lower bound can be shown using the sign rank.

They also used ChatGPT 6 to assist in drafting this manuscript and improve the presentation. They take full responsibility for the content.

\section*{Acknowledgments}

AH thanks Vahid Asadi, Richard Cleve, and Alex May for telling him about non-local quantum computation, Philip Verduyn Lunel for discussions about quantum position verification, Harumichi Nishimura for discussions about CDS, and Daiki Suruga for pointing out \cref{fact:deterministic-one-way}. AH and RM thank Alex May for comments on the manuscript, and Fran{\c c}ois Le Gall for his generous support.

AH is supported by JSPS KAKENHI grant No.~24H00071, 25K24674, 25K24465. RM is supported by the JST CREST grant JPMJCR24I4.

\bibliographystyle{alpha}
\bibliography{ref}

\appendix

\section{Matching upper bound for equality}\label{app:EQ_upperbound}

The logarithmic lower bound for equality is tight for fixed privacy error.

\begin{proposition}[CDS for equality via hashing]
\label{prop:eq-upper}
For every fixed constant $0<\delta<1$, equality on $n$-bit strings, with a one-bit secret \(s\in\{0,1\}\), has a perfectly correct and $\delta$-private CDS protocol using $O(\log n)$ shared-randomness bits and one-bit messages from each party.
\end{proposition}

\begin{proof}
By the existence of asymptotically good binary linear codes, such as Justesen codes~\cite{justesen1972class}, there are constants
\(\alpha\in(0,1)\) and \(c>0\) such that, for every \(n\), one can choose an injective linear encoding
\[
    E:\mathbb F_2^n\longrightarrow\mathbb F_2^{M},
    \qquad M\le cn,
\]
satisfying
\[
    d_H(E(x),E(y))\ge \alpha M
    \qquad\text{whenever }x\neq y.
\]

By padding the codewords with zeros and adjusting $\alpha$ by a constant factor if necessary, we may assume that $M$ is a power of two. This preserves \(M=O(n)\).
 
Choose
\[
    t=
    \left\lceil
        \frac{\log(1/\delta)}{-\log(1-\alpha)}
    \right\rceil.
\]
Since $\delta$ is fixed, $t=O(1)$.
Alice and Bob share independent uniform coordinates $I_1,\ldots,I_t\in[M]$ and define
\[
    h(x):=\bigl(E(x)_{I_1},\ldots,E(x)_{I_t}\bigr)
    \in\{0,1\}^t.
\]
Equal inputs always have equal fingerprints.
For any fixed $x\neq y$, the relative-distance guarantee and independence of the sampled coordinates imply
\[
    \Pr[h(x)=h(y)]
    \leq(1-\alpha)^t
    \leq\delta.
\]

We next apply a perfectly correct and perfectly private CDS protocol for equality to the fingerprints. Independently of the sampled coordinates, Alice and Bob share independent uniform bits
\[
    (R_z)_{z\in\{0,1\}^t}.
\]
All shared randomness is independent of the inputs and the secret and is hidden from the referee.
On inputs \(x,y\) and a one-bit secret \(s\in\{0,1\}\), Alice and Bob send
\[
    A=s\oplus R_{h(x)},
    \qquad
    B=R_{h(y)},
\]
respectively. The referee outputs $A\oplus B$.

If $x=y$, then $h(x)=h(y)$ for every choice of the sampled coordinates, and hence
\[
    A\oplus B=s.
\]
Thus, correctness is perfect.

For privacy, fix $x\neq y$ and let
\[
    p:=\Pr[h(x)=h(y)]\leq\delta.
\]
Conditioned on any choice of coordinates for which $h(x)\neq h(y)$, the masks $R_{h(x)}$ and $R_{h(y)}$ are independent uniform bits. The transcript is therefore uniform on $\{0,1\}^2$, independently of $s$.
Conditioned on a collision, it is uniform over the two pairs whose XOR equals $s$.

Let $P_s$ denote the transcript distribution for secret $s$, let $U_2$ denote the uniform distribution on $\{0,1\}^2$, and let $Q_s$ denote the uniform distribution on $\{(u,v):u\oplus v=s\}$.
Since
\[
    P_s=(1-p)U_2+pQ_s
\]
and
\[
    \|Q_s-U_2\|_1=1,
\]
we obtain
\[
    \|P_s-U_2\|_1
    =
    p\|Q_s-U_2\|_1
    =
    p
    \le \delta.
\]
Thus \(U_2\) provides a simulator independent of the secret, and the protocol is \(\delta\)-private.

Finally, sampling the coordinates and the random table uses
\[
    t\log_2 M+2^t=O(\log n)
\]
shared-randomness bits, since $t$ is constant.
\end{proof}

\section{Direct connection between sign-rank and nondeterministic rank}\label{appendix_nrank-signrank}

We show a relationship between sign-rank and nondeterministic rank. A similar result was shown in \cite[Section~2.1]{goos2025sign}.

\begin{lemma}\label{lem:nrank-signrank}
Let $f:X\times Y\to\{0,1\}$, where $X$ and $Y$ are finite nonempty sets, and let $S_f(x,y):=2f(x,y)-1$.
Then
\[
    \operatorname{signrank}(S_f)
    \le \operatorname{nrank}(f)^2+1,
\]
where $\operatorname{nrank}(f)$ is the minimum rank over $\mathbb{C}$ of a matrix $M$ satisfying $M_{x,y}\neq 0$ if and only if $f(x,y)=1$.
\end{lemma}

\begin{proof}
The claim is immediate if $f\equiv 0$, and assume that $f$ has at least one $1$-input.
Let $M$ be a nondeterministic matrix for $f$ with
\[
    \operatorname{rank}(M)=\operatorname{nrank}(f)=:r.
\]
Define the real matrix
\[
    B:=M\circ\overline{M},
    \qquad
    B_{x,y}=|M_{x,y}|^2,
\]
where $\circ$ denotes the entrywise product and $\overline{M}$ denotes the entrywise complex conjugate of $M$.
Thus $B_{x,y}=0$ when $f(x,y)=0$, and $B_{x,y}>0$ when $f(x,y)=1$.
By the rank inequality for entrywise products\footnote{
For matrices $A=\sum_{i=1}^{k}a_i u_i^\top$ and $C=\sum_{j=1}^{\ell}c_j v_j^\top$, $A\circ C
    =
    \sum_{i=1}^{k}\sum_{j=1}^{\ell}
    (a_i u_i^\top)\circ(c_j v_j^\top)
    =
    \sum_{i=1}^{k}\sum_{j=1}^{\ell}
    (a_i\circ c_j)(u_i\circ v_j)^\top$.
This shows $\operatorname{rank}(A\circ C)\le k\ell = \operatorname{rank}(A)\operatorname{rank}(C)$.
},
\[
    \operatorname{rank}(B)
    \le \operatorname{rank}(M)
       \operatorname{rank}(\overline{M})
    =r^2.
\]
Here the real and complex ranks of $B$ are the same because $B$ is real.

Let $J$ be the all-ones matrix, and choose
\[
    0<\delta<\min_{f(x,y)=1}B_{x,y}.
\]
Then $B-\delta J$ is negative on the $0$-inputs of $f$ and positive on its $1$-inputs. Hence
\[
    \operatorname{sign}(B-\delta J)=S_f.
\]
Consequently,
\[
    \operatorname{signrank}(S_f)
    \le \operatorname{rank}(B-\delta J)
    \le \operatorname{rank}(B)+1
    \le r^2+1,
\]
as claimed.
\end{proof}

\begin{corollary}\label{cor:routing-via-nrank}
For every nonconstant Boolean function $f:X\times Y\to\{0,1\}$,
\[
    \FRzero(f),\ \FRone(f),\ \pFR(f)
    \ge
    \frac18
    \log_2\!\bigl(\max\{1,\operatorname{signrank}(S_f)-1\}\bigr).
\]
In particular, for every $n\ge 1$,
\[
    \FRzero(\IP_n),\ \FRone(\IP_n),\ \pFR(\IP_n)
    \ge \frac{n}{16}-\frac18.
\]
\end{corollary}

\begin{proof}
The lower bounds of Asadi, Culf, and May~\cite{asadi2025rank}
imply
\[
    \FRzero(f)\ge \frac14\log_2\operatorname{nrank}(f),
    \qquad
    \FRone(f)\ge \frac14\log_2\operatorname{nrank}(1-f).
\]
Write $s:=\operatorname{signrank}(S_f)$. Since $S_{1-f}=-S_f$, we have
\[
    \operatorname{signrank}(S_{1-f})=s.
\]
Applying \cref{lem:nrank-signrank} to both $f$ and $1-f$ gives
\[
    \operatorname{nrank}(f)^2\ge s-1,
    \qquad
    \operatorname{nrank}(1-f)^2\ge s-1.
\]
Moreover, both nondeterministic ranks are at least one, because $f$ is nonconstant. Consequently, for $h\in\{f,1-f\}$,
\[
    \log_2\operatorname{nrank}(h)
    \ge \frac12\log_2\!\bigl(\max\{1,s-1\}\bigr).
\]
Combining these inequalities proves the bounds for $\FRzero(f)$ and $\FRone(f)$. The same bound holds for $\pFR(f)$, since
\[
    \pFR(f)\ge \max\{\FRzero(f),\FRone(f)\}.
\]

For inner product, the standard sign-rank lower bound gives
\[
    \operatorname{signrank}(S_{\IP_n})\ge 2^{n/2}.
\]
Using $\max\{1,s-1\}\ge s/2$ for every $s\ge 1$, we obtain
\[
    \frac18\log_2\!\bigl(\max\{1,s-1\}\bigr)
    \ge \frac18(\log_2 s-1)
    \ge \frac18\left(\frac n2-1\right)
    =\frac{n}{16}-\frac18.
\]
\end{proof}

\end{document}